\documentclass[12pt]{amsart}

\usepackage{amsfonts,amsmath,amssymb,amsxtra,amsthm,tikz}
\usepackage{mathrsfs}
\usepackage{times}
\usepackage{anysize}
\usepackage{graphicx}
\usetikzlibrary{arrows,chains,matrix,positioning,scopes}
\makeatletter
\tikzset{join/.code=\tikzset{after node path={%
\ifx\tikzchainprevious\pgfutil@empty\else(\tikzchainprevious)%
edge[every join]#1(\tikzchaincurrent)\fi}}}

\makeatother
\tikzset{>=stealth',every on chain/.append style={join},
         every join/.style={->}}
\tikzstyle{labeled}=[execute at begin node=$\scriptstyle,
   execute at end node=$]

\marginsize{1in}{1in}{1in}{4cm}

\newtheorem{theorem}{Theorem}[section]
\newtheorem{lemma}[theorem]{Lemma}

\theoremstyle{definition}
\newtheorem{definition}[theorem]{Definition}

\newtheorem{proposition}[theorem]{Proposition}

\theoremstyle{remark}

\newcommand{\CO}{\mathcal{O}}
\newcommand{\CT}{\mathcal{T}}

\newcommand{\CF}{\mathcal{F}}

\newcommand{\CM}{\mathcal{M}}

\def\be{\begin{eqnarray}}
\def\ee{\end{eqnarray}}

\def\bfig{\begin{figure}[H] }
\def\efig{\end{figure}}

\def\bc{\begin{center}}
\def\ec{\end{center}}

\def\lm{\limits}

\def\cham{{\frak{C}}}
\def\msp{{\frak{M}}}
\def\matC{{\mathbb{C}}}
\def\matH{{\mathbb{H}}}
\def\matR{{\mathbb{R}}}
\def\matZ{{\mathbb{Z}}}
\def\fa{{\frak{a}}}
\def\fc{{\frak{c}}}
\def\fb{{\frak{b}}}
\def\ft{{\frak{t}}}
\def\fg{{\frak{g}}}
\def\fA{{\frak{A}}}
\def\fL{{\frak{L}}}
\def\fB{{\frak{B}}}

\def\lb{{\bar{\lambda}}}
\def\lm{{\bar{\mu}}}
\def\lnu{{\bar{\nu}}}

\title{Polynomials associated with fixed points on the instanton moduli space }

\author{Andrey Smirnov}
\address{Department of Mathematics, Columbia University, New York USA}
\email{asmirnov@math.columbia.edu}

\begin{document}

\begin{abstract}
  Using the Okounkov-Maulik stable map, we identify the equivariant cohomology of instanton moduli spaces with the space of polynomials on an infinite number of variables. We define the generalized Jack polynomials as the polynomials representing the classes of fixed points under this identification.
  Using the abelianization theorem of Shenfeld we derive the combinatorial formula for the  expansion of generalized Jack polynomials in the basis of Schur polynomials.
\end{abstract}

\maketitle

\section{Introduction}

The purpose of this paper is twofold. The first is to define a special class of polynomial functions associated to the classes of fixed points in the equivariant cohomology  of instanton moduli spaces. We show that these functions provide a natural generalization of symmetric Jack polynomials, in particular they inherit and generalize many of their properties. We call these new functions generalized Jack polynomials.

The generalized Jack polynomials form a basis in the space of colored polynomials, labeled by $r$-tuples of partitions $\lb=[\lambda_1,...,\lambda_r]$. In section \ref{def} we derive some properties of these functions. For $r=1$ we obtain the standard Jack polynomials; the results of this paper both recover some known results and provide new results on Jack polynomials.

Secondly we derive a combinatorial formula for the expansion of generalized polynomials in the basis of Schur  polynomials. These formulae can be used to compute the generalized polynomials and provide a good tool for explicit computations (see Appendix B for examples).

This work was motivated in part by recent papers on the AGT conjecture \cite{Alday:2009aq,Wyllard:2009hg,Mironov:2009by}, which relates correlation functions in $2d$ Liouville conformal field theory with correlators in $4d$ gauge theories. As explained in \cite{Alba:2010qc,Fateev:2011hq,Belavin:2011sw} the AGT conjecture implies the existence of the special basis in the CFT Fock space (known as AFLT states). In this basis the expansion of $2d$ conformal blocks coincide with expansion of $4d$ Nekrasov functions, which establishes the relation between both sides of the conjecture. After bosonization of the Virasoro (or, generally, $W_N$) algebra this basis coincides with the basis of generalized Jack polynomials considered in this paper. Thus, the results of this paper provide explicit combinatorial formula for AFLT states of arbitrary rank.

Another approach to AGT conjecture is based on the free field representations of CFT conformal blocks \cite{Dotsenko:1984nm}. This approach was investigated by several authors \cite{Itoyama:2010ki, Mironov:2010pi,Zhang:2011au,Kanno:2013aha}. In particular in \cite{Morozov:2013rma, Mironov:2013oaa} is was shown that
 the coefficients for the expansion of the conformal blocks in the free field formalism have a form of Selberg integrals of generalized Jack polynomials.
The AGT conjecture, therefore, implies that the $SU(N)$ Selberg integrals of generalized Jack polynomials have nice properties: they completely factorize to  products of linear multiples what gives a powerful generalization of of Kadell's the theorems \cite{Kadell1,Kadell2}. This conjecture was checked explicitly in some cases in \cite{Morozov:2013rma, Mironov:2013oaa}. While we do not discuss the Seleberg integrals of generalized Jack polynomials in this paper, the conjecture mentioned above gives one more example of properties of these functions that generalize some known properties of Jack polynomials. We believe, that all properties of Jack polynomials can be "lifted" to generalized polynomials.

We outline the main idea of the paper. It is well known, that the standard Jack polynomials can be defined as the classes of fixed points in the equivariant cohomology of Hilbert schemes of points on a complex plane. Let $Hilb_n$ be the space of ideals $J\in \matC[x,y] $ with codimension $\dim \matC[x,y]/J=n $.    The scaling action of a torus $C=(\matC^{*})^2$ on polynomials:
$$
(z_1,z_2) \cdot p(x,y)=p(z_1 x, z_2 y)
$$
gives rise to the action of $C$ on the Hilbert scheme $Hilb_n$. Obviously, the ideal  $J\in Hilb_n$ is fixed under $C$-action if and only if
it is generated by monomials $J=\langle x^{a_1} y^{b_1},..., x^{a_m} y^{b_m} \rangle$. Moreover, the fixed ideals of finite codimension $n$ are in one to one correspondence with partitions of $n$. The partition $\lambda=(\lambda_1,\lambda_2,...,\lambda_m )$ with $|\lambda| = \lambda_1+...+\lambda_m=n$
uniquely defines the $C$-fixed ideal $J\in Hilb_n$  given by $J=\langle x^{\lambda_1},x^{\lambda_2}  y,...,x^{\lambda_m}y^{m}, y^{m+1} \rangle$.

Thus, the fixed points of $Hilb_{n}$ are isolated and labeled by partitions of $n$. By Nakajima's construction \cite{Nak1,Nak2}, the localized equivariant cohomology of the Hilbert scheme can be identified with the space of polynomials (the Fock space):
\be
\label{ni}
H^{\bullet}_{C}\Big( \coprod\limits_{n=0}^{\infty} \, Hilb_n  \Big)\simeq \matC[p_1,p_2,....]\otimes \matC(t_1,t_2)
\ee
where $t_1, t_2$ - are the equivariant parameters corresponding to the characters of $C$.
The classes of fixed points $[\lambda]$ provides the canonical basis of this space. Under the identification (\ref{ni}) they coincide with the basis of Jack polynomials $j_{\lambda}(p_k)$ with parameter $\beta=-t_2/t_1$ \cite{Nak3}:
$$
[\lambda] \in H^{\bullet}_{C}\Big( \coprod\limits_{n=0}^{\infty} \, Hilb_n  \Big) \Leftrightarrow j_{\lambda}(p_k) \in \matC[p_1,p_2,....]\otimes \matC(t_1,t_2)
$$

Let $\msp(r,n)$ be the space of rank $r$ -instantons on the complex projective plane with topological charge $n$ (see section \ref{def}). For $r=1$ it is isomorphic to the Hilbert scheme of points $\msp(1,n)=Hilb_n$. Similarly, the space
$\msp(r,n)$ comes with the natural action of torus $C=(\matC^*)^{r+2}$ on it. The fixed set $\msp(r,n)^{C}$ consist of isolated fixed points labeled by $r$-tuples of partitions, $\bar{\lambda}=[\lambda_1,...,\lambda_{r}]$, with the total number of boxes given by  $|\bar{\lambda}|=|\lambda_1|+...+|\lambda_r|=n$. After localization the classes of fixed points $[\bar{\lambda}]$ form a basis in cohomology of $\msp(r,n)$.

As with (\ref{ni}) we use a stable map introduced by Okounkov and Maulik in \cite{Maulik:2012wi} to  identify the cohomology of instanton moduli spaces with  polynomials on an infinite number of variables $p_{n}^{(i)}$ $n\in {\mathbb{Z}}_{+}$ colored by index $i=1...r$:
\be
\label{oi}
H^{\bullet}_{C}\Big( \coprod\limits_{n=0}^{\infty} \, \msp(r,n)  \Big) \stackrel{Stab}{\longleftarrow} \matC[p^{(1)}_1,p_2^{(1)},....p^{(2)}_1,p_2^{(2)},...p^{(r)}_1,p_2^{(r)}...]\otimes \matC(t_1,t_2,u_1,...,u_r)
\ee
When $r\neq 1$ there are however two canonical nonequivalent choices of stable maps $Stab_{\cham^{\prime}_{+}}$ and $Stab_{\cham^{\prime}_{-}}$,   related to the choice of a chamber in $C$. Thus, in this case, we have two distinct identifications (\ref{oi}) and we define two sets of polynomials $J_{\lb}(p_{n}^{k})$ and $J_{\lb}^{*}(p_{n}^{k})$ (labeled now by $r$-tuples of partitions) by:
$$
Stab_{\cham^{\prime}_{+}}( J_{\lb} )=Stab_{\cham^{\prime}_{-}}( J_{\lb}^{*} ) =[\lb] \in H^{\bullet}_{C}\Big( \coprod\limits_{n=0}^{\infty} \, \msp(r,n)  \Big)
$$
When $r=1$ our definition reproduces the standard Jack polynomials $J_{\lambda}=J^{\ast}_{\lambda}=j_{\lambda}$ (see for example, explicit formulae in appendix B). The space $(\ref{oi})$ comes with a natural scalar product, and we will show that $J_{\lb}$ and $J_{\lb}^{*}$ give two bases in the space of polynomials which are dual with respect to this scalar product. In section \ref{def} we derive several properties of generalized Jack polynomials: explicit formulae for scalar products, Cauchy identities, and some symmetries corresponding to change of a chamber. By proposition \ref{efun} below the generalized Jack polynomials can also be defined as a basis of eigenfunctions for some integrable system ($r$-interacting Calogero-Moser systems).
For $r=1$ all these results specialize to well known properties of Jack polynomials.

Let $s_{\lambda}$ be the standard Schur polynomials giving a basis in $\matC[p_1,p_2,...]$. Let us consider a basis in (\ref{oi}) labeled by $r$-tuples of partitions and given by a product of Schur functions:
$$
s_{\lb}=s_{\lambda_{1}}(p_k^{(1)}) s_{\lambda_{2}}(p_k^{(2)})...s_{\lambda_{r}}(p_k^{(r)}) \in \matC[p^{(1)}_1,p_2^{(1)},....p^{(2)}_1,p_2^{(2)},...p^{(r)}_1,p_2^{(r)}...]
$$
Similarly let $s^{*}_{\lambda}(p_k)=s^{*}_{\lambda}(-t_2/t_1 p_k)$ and $s^*_{\lb}=s^*_{\lambda_{1}}(p_k^{(1)}) s^*_{\lambda_{2}}(p_k^{(2)})...s^*_{\lambda_{r}}(p_k^{(r)})$.

We use the abelianization technique developed by D. Shenfeld in his thesis \cite{Shenf1} to prove the following theorem:

\begin{theorem}
\label{thmone}
The generalized Jack polynomials have the following expansions in the basis of Schur functions:
$$
J_{\lb}=\sum\limits_{|\lm|=|\lb|}\, T_{\lb,\lm} s_{\lm}, \ \ \ J^*_{\lb}=\sum\limits_{|\lm|=|\lb|}\, T^*_{\lb,\lm} s^*_{\lm},
$$
with the coefficients given by the following  combinatorial formulas:
\begin{small}
$$
 {{T}}_{\bar \lambda,\bar\mu}=\dfrac{1}{\frak{z}({\bar \mu})} \sum\limits_{\sigma\in S_{|\lb|}}\,
\dfrac{\prod\limits_{\Box_1,\Box_2=1}^{|\lb|} \, \Big(\varphi^{\sigma\bar\lambda}_{\Box_2} -
 \varphi^{\sigma\bar\lambda}_{\Box_1} +t_1 \Big)^{\ast\langle \rho^{\bar\mu}_{\Box_2}-\rho^{\bar\mu}_{\Box_1}+t_1 | \cham_{-}\rangle} \prod\limits_{d=1}^{r}\prod\limits_{\Box=1}^{|\lb|}\Big(u_d -\varphi^{\sigma\bar\lambda}_{\Box} \Big)^{\ast \langle u_d-\rho^{\bar\mu}_{\Box} | \cham_{-} \rangle}}
 {\prod\limits_{\Box_1>\Box_2 \atop \rho^{\bar\mu}_{\Box_1}\neq \rho^{\bar\mu}_{\Box_2} } \Big( \varphi^{\sigma\bar\lambda}_{\Box_1}-\varphi^{\sigma\bar\lambda}_{\Box_2} \Big)
 \Big( \varphi^{\sigma\bar\lambda}_{\Box_1}-\varphi^{\sigma\bar\lambda}_{\Box_2} +\hbar\Big)}
$$
$$
\ \ \ \  {{T}}^{*}_{\bar \lambda,\bar\mu}=\dfrac{1}{\frak{z}({\bar \mu})} \sum\limits_{\sigma\in S_{|\lb|}}\,
\dfrac{\prod\limits_{\Box_1,\Box_2=1}^{|\lb|} \, \Big(\varphi^{\sigma\bar\lambda}_{\Box_2} -
 \varphi^{\sigma\bar\lambda}_{\Box_1} +t_1 \Big)^{\ast\langle \rho^{\bar\mu}_{\Box_2}-\rho^{\bar\mu}_{\Box_1}+t_1 | \cham_{+}\rangle} \prod\limits_{d=1}^{r}\prod\limits_{\Box=1}^{|\lb|}\Big(u_d -\varphi^{\sigma\bar\lambda}_{\Box} \Big)^{\ast \langle u_d-\rho^{\bar\mu}_{\Box} | \cham_{+} \rangle}}
 {\prod\limits_{\Box_1<\Box_2 \atop \rho^{\bar\mu}_{\Box_1}\neq \rho^{\bar\mu}_{\Box_2} } \Big( \varphi^{\sigma\bar\lambda}_{\Box_1}-\varphi^{\sigma\bar\lambda}_{\Box_2} \Big)
 \Big( \varphi^{\sigma\bar\lambda}_{\Box_1}-\varphi^{\sigma\bar\lambda}_{\Box_2} +\hbar\Big)}
$$
\end{small}
\end{theorem}
Let us briefly explain the notations in \ref{thmone} (details can be found in section \ref{absec}). Given an $r$-tuple of partitions  $\lb$ we assume that its boxes ordered as in figure \ref{ordfig} below. The sums run over the permutations of boxes in $\lb$. The products in both formulas run over boxes in $\lb$. Given a box $\Box \in \lb =[\lambda_1,...,\lambda_r]$ which belongs to the partition $\lambda_k$ and has standard coordinates $x_{\Box}$ and $y_{\Box}$ we denote:
\be
\label{croot}
\varphi_{\Box}^{\lb}=u_k+x_{\Box}  t_1+y_{\Box}  t_2,
\ee
and its projection to the torus $t_2+t_1=0$:
\be
\rho_{\Box}^{\lb}=u_k+x_{\Box} t_1-y_{\Box} t_1.
\ee
We use Shenfeld's notations:
\be
\label{Shenfnot}
(x)^{* a}=\left\{\begin{array}{ll}
x & a>0\\
\hbar-x & a<0\\
1 & a=0
\end{array}\right.
\ee
with $\hbar=t_1+t_2$. Let us consider two opposite chambers:
$$
\cham_{+}=\{ u_i-u_{i+1} \gg t_1>0\}, \ \ \cham_{-}=\{ u_{i+1}-u_{i} \gg t_1<0 \}
$$
We denote by $\langle x, \cham \rangle$ the sign of the character $x$ on the chamber $\cham$. Finally $\frak{z}(\lm)$ is the numerical factor given explicitly by formula (\ref{fzet}).

Note, even at $r=1$ this gives nontrivial formula for coefficients in the expansion of standard Jack polynomials in the Schur polynomials, which gives inverse formula for one computed in \cite{Shenf1}.

Most of the results of this article have a straightforward generalization to the equivariant K-theory. This provides analogous combinatorial formulas for Macdonald polynomials and their generalized version. We plan to publish these in a separate paper.

\section*{Acknowledgments}
We are deeply grateful to D. Shenfeld for sharing and numerous explanations of his results on abelianization of stable envelopes. We  would also like to thank M. Mcbreen, A. Okounkov and V. Pal for useful discussions. This work was supported in part by RFBR grant 12-01-00482.

\section{Definition of generalized Jack polynomials \label{def}}
 Let $\msp(r,n)$ be the moduli space of framed rank $r$ torsion free sheaves $\CF$ on ${\mathbb{P}}^2$ with fixed Chern class $c_{2}(\CF)=n$. A framing of $\CF$ is the choice of isomorphism
\be
\label{frm}
 \left.\CF\right|_{{\mathbb{P}}} =\CO^{\oplus r}_{{\mathbb{P}}}
\ee
where ${\mathbb{P}}\subset {\mathbb{P}}^2$ is considered as a line at infinity. The space has a natural action by the group $G=GL(r)\times GL(2)$
where $GL(2)$ acts on ${\mathbb{P}}^2$ preserving infinity line, and $GL(r)$ acts by changing framing (\ref{frm}).  Let $C$ is the maximal torus of $G$ and
$A$ be the maximal torus of $GL(r)$ such that $A\subset C$. Note that $\msp(1,n)=Hilb_{n}$ - the Hilbert scheme of $n$ points on $\matC^2$ \cite{Nak2}.  We will refer to $\msp(r,n)$ as the instanton moduli space. By definition, a sheaf $\CF$ is fixed under action of framing torus $A$ if it splits
to a sum of rank one sheaves $\CF= J_1\oplus...\oplus J_r$ therefore we have:
$$
\msp^{A}(r,n)=\coprod \limits_{n_1+...+n_r=n}\, Hilb_{n_1}\times...\times Hilb_{n_r}
$$
The fixed set of bigger torus $C$ is obviously the fixed set of two-dimensional torus $C/A$ on $\msp^{A}(r,n)$. Torus $C/A$ acts on the Hilbert schemes as described in the introduction, thus fixed set of $C$ is given by isolated fixed points labeled by $r$-tuples of partitions with total number
of $n$-boxed:
$$
\msp^{C}(r,n)=\{ \lb=[\lambda_1,...,\lambda_r] : |\lb|=\sum_i |\lambda_i|=n \}
$$
We use the Nakajima construction of Heisenberg algebra \cite{Nak2} to identify cohomology of Hilbert schemes with "boson Fock space":
$$
H^{\bullet} \Big( \coprod\limits_{n=0}^{\infty} \, Hilb_{n} \Big)\simeq \matC[p_1,p_2,...]=\textsf{F}
$$
Define $\msp(r)=\coprod_{n=0}^{\infty} \msp(r,n)$ then from above we have $H^{\bullet}(\msp^{A}(r)) \simeq \textsf{F}^{\otimes r}$. In particular we have:
$$
H_{C}^{\bullet} \Big( \msp^{A}(r)  \Big)\simeq \matC[p^{(1)}_1,p_2^{(1)},....p^{(2)}_1,p_2^{(2)},...p^{(r)}_1,p_2^{(r)}...]\otimes \matC(t_1,t_2,u_1,...,u_r)
$$
where $u_1,..u_r,t_1, t_2$ - the equivariant parameters corresponding to characters of $C$ (  $u_i$ are the characters of $A$ ).

The action of $A \subset C$ on the moduli space provides corresponding Lie algebras $\fa \subset \fc$ with chamber
decompositions (Appendix A). The chambers in $\fa$ are the standard Weyl chambers of $GL(r)$.
We consider:
\be
\label{chambersp}
\cham^{\, \prime}_{+}=\{u_1>u_2>...>u_r\}, \ \ \cham^{\, \prime}_{-}=\{u_1<u_2<...<u_r\}
\ee
two fixed, opposite chambers in $\fa$.
There are infinitely many chambers in $\fb$ \cite{Maulik:2012wi, Smirnov:2013hh}. We fix two opposite chambers in $\fb$ defined uniquely by conditions
$\cham^{\, \prime}_{\pm}\subset \overline{\cham}_{\pm}$ i.e. $\cham^{\, \prime}_{\pm}$
are the walls of these chambers and $\pm t_1 >0$. Explicitly we have:
\be
\label{chambers}  \cham_{+}=\{ u_i-u_{i+1} \gg t_1>0\}, \ \ \cham_{-}=\{ u_{i+1}-u_{i} \gg t_1<0 \}
\ee
Given a chamber $\cham \subset \fa$ we can consider the Okounkov-Maulik stable map (see appendix A):
$$
H^{\bullet}_{C}\Big( \msp^{A}(r)  \Big) \stackrel{Stab_{\cham_{\pm}}}{\longrightarrow} H_{C}^{\bullet}\Big( \msp(r)  \Big)
$$
After localization these maps become isomorphisms and the equivariant classes of fixed points $[\lb]\in H_{C}^{\bullet}\Big( \msp(r)  \Big)$
give a basis in cohomology.

\begin{definition}
Let us consider stable maps defined by chambers $\cham_{\pm}$:
\be
\label{defn}
\matC[p^{(1)}_1,p_2^{(1)},....p^{(2)}_1,p_2^{(2)},...p^{(r)}_1,p_2^{(r)}...]\otimes \matC(t_1,t_2,u_1,...,u_r) \stackrel{Stab_{\cham_{\pm}}}{\longrightarrow} H_{C}^{\bullet}\Big( \CM(r)  \Big)
\ee
The generalized Jack polynomials $J_{\lb}$ and $J_{\lb}^{*}$ are defined as polynomials corresponding to $C$ fixed points under the above identifications:
$$
Stab_{\cham^{\prime}_+}\Big( J_{\lb} \Big)=Stab_{\cham^{\prime}_-}\Big( J_{\lb}^{*} \Big)= [\lb] \in H_{C}^{\bullet}\Big( \CM(r)  \Big)
$$
We conclude this sections by describing a few simple properties of generalized Jack polynomials.
\end{definition}

\begin{proposition}
Define grading in polynomial ring (\ref{defn})  by $\deg(p_{k}^{(i)})=k$. Then
$$
\deg(J_{\lb})=\deg(J^{*}_{\lb})=|\lb|
$$
\end{proposition}

\begin{proof}
In the Nakajima construction the cohomological degree of $p_n$ is $2n$. The proof follows.
\end{proof}

\begin{proposition}
For $r=1$ we have $J_{\lambda}(p_k)=J_{\lambda}^{*}(p_k)=j_{\lambda}(p_k)$ where $j_{\lambda}(p_k)$ is the standard Jack polynomial.
\end{proposition}

\begin{proof}
For $r=1$ the stable maps $Stab_{\cham_\pm}$ are trivial. The classes of fixed point $[\lambda] \in H^{\bullet}_{C}( \coprod_n  Hilb_{n})$ are given by standard Jack polynomials \cite{Nak3}.
\end{proof}
Let us consider the standard pairing on cohomology:
$$
H^{\bullet}_{C}\Big(\msp^{A}(r)\Big)\times H^{\bullet}_{C}\Big(\msp^{A}(r)\Big) \stackrel{\langle , \rangle}{\longrightarrow}  H^{\bullet}_{C}(pt)
$$
given by:
$$
\langle \alpha,\beta \rangle \rightarrow \int\limits_{\msp^{A}(r)}\, \alpha \cup \beta
$$
and the integral denotes the equivariant residue. By definition, the classes of stable points form an orthogonal basis with respect to this product.
For example for $r=1$, when classes of fixed points are represented by Jack polynomials $j_{\lambda}$ we obtain:
\be
\label{intef}
\langle j_{\lambda}, j_{\mu} \rangle= \int\limits_{Hilb_{n}} [\lambda] [\mu] = \delta_{\lambda,\mu}  e_{\lambda,\mu}(0)
\ee
where $e_{\lambda,\mu}(0)$  is the character of tangent space to the point $\lambda$, given explicitly by formula (\ref{efor}). This scalar product is well known, (and usually used to define the Jack polynomials). Explicitly, the scalar product of two polynomials $r(p_{n}), s(p_{k})$ can be computed as follows:
\be
\label{spdif}
\langle r(p_{n}), s(p_{k}) \rangle= \left. \Big( r(-n t_1/t_2 \frac{\partial}{\partial p_k}) s(p_{k}) \Big) \right|_{p_k=0}
\ee
where we imply that the polynomial $r(p_{n})$ is substituted by differential operator $p_n \rightarrow -n t_1/t_2 \partial_{p_n}$ and applied to
the polynomial $s(p_{k})$.

\begin{proposition}
The basis $J^{*}_{\lb}$ is dual to $J_{\lb}$ with respect to this scalar product:
\be
\langle J_{\lb},  J_{\lm}^{*} \rangle =  E_{\lb,\lm}
\ee
with $E_{\lb,\lm}=\prod\limits_{i,j=1}^{r} \, e_{\lambda_{i},\mu_{j}}(u_i-u_j)$ and
\be
\label{efor}
e_{\lambda,\mu}(u)=\prod\limits_{\Box\in \lambda} ( u+(a_{\lambda}(\Box)+1)t_1-l_{\mu}(\Box) t_2 )
\prod\limits_{\Box \in \mu} \,( u-a_{\mu}(\Box)t_1+(l_{\lambda}(\Box)+1) t_2 )
\ee
in particular $E_{\lb,\lm} =0$ if $\lb\neq\lm$.
\end{proposition}
\begin{proof}
The proof follows from the defining properties of stable map (Appendix B, theorem \ref{stabth} ). Let us compute the integral (\ref{intef}) using Atiyah-Bott localization to fixed points. If $\lb=\lm$ then by the first property, the only point that contribute is $\lb$. By the second property the answer is given by the Euler class of the tangent bundle $T \msp(r,n)$ evaluated by $\lb$. Calculation gives:
\be
\left.e(T \msp(r,n))\right|_{\lb}=\prod\limits_{i,j=1}^{r} \, e_{\lambda_{i},\mu_{j}}(u_i-u_j)
\ee
with $e_{\lambda,\mu}$ as in the theorem. If $\lb\neq\lm$ then by the first and the third property in theorem \ref{stabth}, the residue at each point is of degree less than $ \dim  \msp(r,n) $ and therefore is zero.
\end{proof}

\begin{proposition}
We have analog of Cauchy identity:
\be
\label{cauchy1}
\sum\limits_{\lb}\, \left( \dfrac{ J_{\lb}(p_{n}^{(i)}) J^{*}_{\lb}(q_{n}^{(i)})}{ E_{\lb,\lb}} \right)=
\exp\Big(-\frac{t_2}{t_1} \sum\limits_{n=1}^{\infty} \, \frac{p^{(1)}_n q^{(1)}_n+...+p^{(r)}_n q^{(r)}_n}{n} \Big)
\ee
\end{proposition}
\begin{proof}
Denote the left side of this identity by
$$
Id=\sum\limits_{\lb}\,  \dfrac{ J_{\lb} J^{*}_{\lb} }{ E_{\lb,\lb} }
$$
as $e_\lb=\prod\limits_{i,j=1}^{r} \, e_{\lambda_{i},\lambda_{j}}(u_i-u_j)$ is the norm of $J_{\lb}$, this is an identity operator in the sense that:
$$
\langle Id, v \rangle=\sum\limits_{\lb}\,  \dfrac{ J_{\lb}}{ E_{\lb,\lb}} \, \langle J^{*}_{\lb} , v \rangle=v
$$
for any class $v$. Therefore, by (\ref{spdif}) its enough to note that after substitution $q_n = \partial_{q_n}$ the left side acts as identity on any polynomial:
$$
\left.\exp\Big(\sum\limits_{n=1}^{\infty} \, p^{(1)}_n \partial{q^{(1)}_n}+...+p^{(r)}_n \partial{q^{(r)}_n} \Big) v(q_1,...,q_r)\right|_{q_i=0}=v(p_1,...,p_r)
$$
\end{proof}

\begin{proposition}
We have the symmetry of $u$-characters:
\be
\label{sym1}
J^{u_1,..,u_r}_{\lambda_1,...,\lambda_r}(p_n^{(1)},...,p_n^{(r)})= (-1)^{|\lb|}
J^{* {u_r,..,u_1} }_{\lambda_r,...,\lambda_1}(p_n^{(r)},...,p_n^{(1)})
\ee
and similar symmetry of $t$-characters:
\be
\label{sym2}
J^{t_1,t_2}_{\lambda_1,...,\lambda_r}(p_n^{(i)})=J^{\, t_2,t_1}_{\lambda^{\prime}_1,...,\lambda_r^{\prime}}(p_n^{(i)} t_2/t_1) \ \ \ J^{* t_1,t_2}_{\lambda_1,...,\lambda_r}(p_n^{(i)})=J^{* \, t_2,t_1}_{\lambda^{\prime}_1,...,\lambda_r^{\prime}}(p_n^{(i)} t_2/t_1)
\ee
where $\lambda^{\prime}$ denotes the transposed diagram.
\end{proposition}
\begin{proof}
For (\ref{sym1}) enough to note that the changing the chamber $ \cham^{\prime}_{+} \leftrightarrow \cham^{\prime}_{-} $  is the same as change of the order of equivariant parameters $(u_1,..,u_r)\leftrightarrow (u_r,..,u_1)$ and change of the order of fixed components of $\msp^{A}(r)$
 which correspond to $(p_n^{(1)},...,p_n^{(r)})\leftrightarrow (p_n^{(r)},...,p_n^{(1)})$. Similarly for (\ref{sym1}), the substitution $t_1 \leftrightarrow t_2$ changes the order on boxes corresponding to transposition of all Young diagrams $\lambda_i \leftrightarrow \lambda^{\prime}$.
\end{proof}

\begin{proposition}
We have the following Cauchy identity:
\be
\prod\limits_{k=1}^{r} \prod\limits_{i=1}^{n_{k}} \prod\limits_{j=1}^{m_{k}}\,(1-x_i^{(k)} y_i^{(k)} ) =\sum\limits_{\lb}\, \left( \dfrac{ J^{t_1,t_2}_{\lb}(x_{n}^{(i)}) J^{*,t_2,t_1}_{\lb^{\prime}}(y_{n}^{(i)})}{ \prod\limits_{i,j=1}^{r} \, e_{\lambda_{i},\lambda_{j}}(u_i-u_j)} \right)
\ee
where we  used change of variables $p_k^{(i)}=\sum_m (x^{(i)}_{m})^{k}$ and $\lb^{\prime}=[\lambda_{1}^{\prime},...,\lambda_{r}^{\prime}]$.
\end{proposition}
\begin{proof}
Substituting $q_n^{(i)}=t_1/t_2 q_n^{(i)}$ in (\ref{cauchy1}) and applying (\ref{sym2}) we obtain the result.
\end{proof}

The standard Jack polynomials degenerate to Schur polynomials at $t_1+t_2=0$ similarly we for generalized polynomials we have
\begin{proposition}
Let $t_1+t_2=0$ then the generalized Jack polynomials degenerate to a product of Schur polynomials:
$$
\left.J_{\lb}(p_{n}^{i})\right|_{t_1+t_2=0}=(-1)^{|\lb|} \dfrac{\prod\limits_{1\leq i<j<\leq r}\, e_{\lambda_{i},\lambda_{j}}(u_i-u_j)}{\prod\limits_{\Box \in \lb} \,\textrm{ hook}(\Box)} s_{\lambda_{1}}(p_{n}^{(1)})...s_{\lambda_{r}}(p_{n}^{(r)})
$$
\end{proposition}
\begin{proof}
At $t_1+t_2=0$ the stable maps is diagonal, i.e. it sends classes of the fixed points to the fixed points modulo a multiple.  Thus, classes $J_{\lb}$ coincide with classes of fixed points on $\msp^{A}$ given by  $s_{\lambda_1}...s_{\lambda_{r}}$ at $t_1+t_2=0$. The diagonal elements of the stable map are given by the Euler class of positive half of the normal bundle to the fixed component. Calculation of this character gives the coefficient.
\end{proof}
The standard Jack polynomials may be defined as common eigenvectors of some infinite family of commuting hamiltonians, known as
the  trigonometric Calogero-Moser integrable system. The hamiltonian of this system has the form:
$$
\textsf{H}^{(r)}=\sum\limits_{m,n=1}^{\infty}\, t_1 \alpha_{-n-m}^{(r)} \alpha_{n}^{(r)} \alpha_{m}^{(r)}-t_2 \alpha_{-n}^{(r)} \alpha_{-m}^{(r)} \alpha_{m+n}^{(r)}+\sum\limits_{n=1}^{\infty}\, (u_r+ \hbar (n-1)/2) \alpha^{(r)}_{-n}  \alpha^{(r)}_{n}
$$
where $\alpha^{(r)}_{n}$ are the generators of Heisenberg algebra:
$$
\alpha^{(r)}_{n}= \left\{ \begin{array}{ll} p^{(r)}_{n} & n<0\\
 n \partial_{p^{(r)}_n} & n>0\\
 0 & n=0  \end{array}  \right.
$$
Let us consider the hamiltonian describing $r$ interacting Calogero-Moser systems:
\be
\label{ham}
\hat{\textsf{H}}= \sum\limits_{i=1}^{r}\, \textsf{H}^{(i)}+ \sum\limits_{1 \leq i<j\leq r}\, \textsf{H}^{(i,j)}
\ee
with interaction term given by:
$$
\textsf{H}^{(i,j)}=\hbar \sum\limits_{k=1}^{\infty}\, (-1)^{k (i-j)} k \alpha_{-k}^{(j)} \alpha_{k}^{(i)}
$$

\begin{proposition} \label{efun}
The eigenvectors of hamiltonian (\ref{ham}) are given by generalized Jack polynomials with the following eigenvalues:
$$
\hat{\textsf{H}} \, J_{\lb} =\Big( \sum\limits_{\Box\in \lb}\, \varphi^{\lb}_{\Box} \Big)\, J_{\lb}
$$
\end{proposition}
\begin{proof}
Under identification (\ref{defn}) (for chamber $\cham_{+}$) the operator (\ref{ham}) corresponds to operator of multiplication by
the first Chern class $c_{1}( \CT )$ tautological bundle $\CT$ over $\msp(r)$ \cite{Maulik:2012wi,Smirnov:2013hh}. In equivariant multiplication by $c_{k}(\CT)$ is always diagonal in the basis of fixed points, with eigenvalues given by $e_{k}(x_1,..,x_k)$ there $e_k$ is $k$-th elementary symmetric function and $x_i$ - Chern roots of $\CT$.
In our case $x_i=\varphi^{\lb}_{\Box_i}$ what gives the eigenvalue.
\end{proof}

\section{ Abelianization of instanton moduli space  \label{absec} }
The aim of this section is to proof the theorem \ref{thmone}. As a byproduct we will also obtain
inverse formulas given by theorem \ref{invth}. We start with a short outline on geometry of hypertoric varieties. The fuller treatment
may be found in \cite{Shenf1,Proud1} and references therein.
\subsection{Hypertoric varieties \label{hvar}}
Let us consider a torus $T^{m}=(\matC^{*})^m$ with a canonical symplectic action on $T^{*}\matC^m$. Let $\ft^m=Lie(T^{m})$ be its Lie algebra and
$\mu_m: T^{*}\matC^m \rightarrow (\ft^m)^{*}$ be the moment map given explicitly by:
$$
\mu_m(z,w)=(z_1w_1,...,z_m w_m)
$$
Let $T^{k}\subset T^{m}$ be algebraic torus. Denote the quotient by $T^{d}=T^m/T^k$ ($d=m-k$) and by $\ft^k$,~$\ft^d$ the corresponding Lie algebras.
We have exact sequence:
$$
0\rightarrow \ft^k \stackrel{i}{\rightarrow} \ft^m \stackrel{j}{\rightarrow} \ft^d \rightarrow 0
$$
and its dual:
$$
0\rightarrow (\ft^d)^{*} \stackrel{j^*}{\rightarrow} (\ft^m)^{*} \stackrel{i^*}{\rightarrow} (\ft^k)^{*} \rightarrow 0
$$
Denote the moment map for $T^k$ action by $\mu_k= i^{\ast} \mu_{m}$. Fix a character $\theta$ of $T^k$, and define the \textit{hypertoric variety} as the following GIT quotient:
\be
\msp=\mu_{k}^{-1}(0)/\!\!/_{\theta}T^k
\ee
The action of $T^m$ on $T^{*}\matC^m$ induces symplectic action of $T^d$ on $\msp$.  We will need, however, an action of a bigger torus on $\msp$:
$\textbf{G}=T^d\times \matC^{*}$. The action of $\matC^{*}$ is induced from its action on $T^{*}\matC^m$ by dilating the fibers. This action scales the the canonical symplectic form $\omega$ on $T^{*}\matC^m$. We denote by $\hbar$ the corresponding character of $\matC \omega$.

Geometry of hypertoric varieties is encoded in terms of hyperplane arrangements. Let us denote $(\ft^{k}_{\matZ})^{*}$ and by $(\ft^{k}_{\matR})^{*}=(\ft^{k}_{\matZ})^{*} \otimes \matR$ the integral and real part of coalgebras (similarly for $\ft^m$ and $\ft^d$).
Let us denote by $e_{i}$ generators of $\ft^m_{\matZ}$ and by $a_i=j(e_{i})$ their images in $\ft^d_{\matZ}$. Let $\hat{\theta}$ be some lift of $\theta$
to $(\ft^m)^{*}$ with coordinates $\hat{\theta}_k$.
The \textit{hyperplane arrangement} is defined as a collection $\fA=\{H_1,...,H_m\}$ of hyperplanes in $(\ft^{d}_{\matR})^{*}$ defined by the equations:
$$
H_{k}=\{ x\in (\ft^{d}_{\matR})^{*}: \langle x,a_k\rangle+\hat{\theta}_k=0 \}
$$
considered with orientation. The orientation means that each hyperplanes divides $(\ft^{d}_{\matR})^{*}$ into a positive and negative half-spaces:
$$
\matH_{k}^{\pm}=\{ x\in (\ft^{d}_{\matR})^{*}: \pm (\langle x,a_k\rangle+\theta_k)>0 \}
$$

The $\textbf{G}$-equivariant cohomology of the hypertoric variety $\msp$ have the following simple description in terms of hyperplane arrangement.
Given a plane $H_{i}$, consider the corresponding character $e_i^{*} \in (\ft^{m})^{*}$. Let $T^{*} \matC^{m} \times \matC$ be the trivial equivariant bundle with the action of $T^m$ in the fiber defined  by $e_i^{*}$. This induces the $\textbf{G}$-equivariant line bundle $L_{i}$ on $\msp$.
Note that there is one to one correspondence between basis elements $e_i\in \ft^{m}$  and the line bundles $L_{i}$. To abuse the notations we denote by the same symbol $e_{i}=e(L_i)$ the Euler class of these line bundles. By construction this character corresponds to the divisor $z_i=0$. Similarly, the dual divisor $w_i=0$ corresponds to the class $\hbar -u_i$ (the shift by $\hbar$ is due to scaling action of $\textbf{G}$ on symplectic form).
\begin{theorem} (\cite{Proud2})
\label{Pth}
The group $H_{\textbf{G}}^{\bullet}( \msp )$ is generated by $(e_1,...,e_m, \hbar)$ with the relations given by:
$$
\prod_{i \in S } e_{i}^{* h_i}=0
$$
for each circuit of hyperplanes $S$.
\end{theorem}
Above, the circuit $S$ is defined as a minimal collection of half-spaces $S=\{ \matH_{k}^{\pm} \}$ such that their intersection is trivial:
$$
\bigcap \limits_{i \in S} \, \matH_{i}^{\pm} =\emptyset
$$
and for $\matH_{i}^{\pm}$ we define a vector $h_{i}=\pm 1$. We also use Shenfeld's notations (\ref{Shenfnot}).

\subsection{Stable basis in hypertoric case}
Given a subset of hyperplanes $\fB \subset \fA$ with nonempty intersection, we consider the hypertoric subvariety $\msp_{\fB} \subset \msp$. The subvariety $\msp_{\fB}$ is defined by
its hyperplane arrangement: let us consider the intersections of all hyperplanes in  $\fB$. This space is not trivial by definition. The intersection of this space with complementary hyperplanes $\fA \setminus \fB$ defines the hyperplanes arrangement in it, and the last defines $\msp_{\fB}$.

Let $T_{\fB} \subset T^{d}$ is a subtorus generated by normals $a_{i}$ to hyperplanes $H_{i}$  for $i\in \fB$. The action of this torus fixes the subvariety $\msp_{\fB}\subset \msp^{T_{\fB}}$ and preserves symplectic form. Given a chamber $\cham \subset \ft_{\fB} =Lie(T_{\fB})$, we have a stable map:
$$
Stab_{\cham}: \, H^{\bullet}_{\textbf{G}}\Big( \msp_{\fB} \Big) \rightarrow H^{\bullet}_{\textbf{G}}\Big( \msp \Big)
$$
The stable envelopes have a nice description in terms of canonical classes $e_i$, we have:
\begin{theorem} (\cite{Shenf1})
For a class  $\gamma\in  H^{\bullet}_{\textbf{G}}\Big( \msp_{\fB} \Big)$ we have
$$
Stab_{\cham}(\gamma) = \gamma \, \prod\limits_{i\in \fB} \, e_{i}^{* \langle \alpha_i,\,\cham\rangle}
$$
where $\alpha_i \in (\ft_{\fB})^{*}$ is the basis dual to $a_i$.
\end{theorem}
By  $\langle \alpha_i,\,\cham\rangle$ we mean the sign of $\alpha_i$ on a chamber $\cham$ i.e. the sign of $\langle \alpha_i,\,\sigma\rangle$ for any cocharacter $\sigma \in \cham$. For our purposes, it will be convenient to rewrite this formula in the other form. First, note that a pairing $ \langle \alpha_i,\,\sigma\rangle$ is the same as considering the restriction of a character to the torus  $T_{\fB}$.
Indeed, the inclusion $T_{\fB}\subset T$ induces the map on coalgebras:
$$
H^{2}_{T}(\msp)  \stackrel{\kappa}{\longrightarrow}  H^{2}_{T_{\fB}}(\msp)\simeq \ft_{\fB}^{*}
$$
then, $ \langle \alpha_k,\,\sigma\rangle=  \langle \kappa(e_k),\,\sigma\rangle$. Moreover, if $\alpha_k$ is such that $k\not \in \fB$, then
$\langle i(e_k),\,\sigma\rangle=0$. Therefore, it is convenient to rewrite the above formula in the form of a product over \textit{all} hyperplanes in the arrangement:
\be
\label{stb}
Stab_{\cham}(\gamma) = \gamma \, \prod\limits_{i=1}^{m} \, e_{i}^{* \langle \kappa(e_i),\,\cham\rangle}
\ee

\subsection{Abelianization of hyperk\"{a}hler quotient}


Assume the reductive group $G$ acts on $V=\matC^m$. It induces the hamiltonian action on $T^{*}V$ and let $\mu_{G}$ be the corresponding moment map.
For a fixed character $\theta \in  \fg^{*} $ we consider the hyperk\"{a}hler quotient:
$$
\msp=\mu_{G}^{-1}(0)/\!\!/_{\theta} G
$$
Let $T\subset G$ be maximal torus and $\ft =Lie(T)$ its Lie algebra. Let $\pi_{T}: \fg^{*} \rightarrow \ft^{*}$ be a projection map on coalgebras, and
$\mu_{T}=\pi_{T} \circ \mu_{G}$. The abelianization of $\msp$ is the following hypertoric variety ( the quotient of $T$ -semistable points ):
$$
\bar{\msp}=\mu_{T}^{-1}(0)/\!\!/_{\theta} T
$$
Assume that a torus $A$ has a right action on $T^{*} V $, commuting with the action of $G$. This induces the action of $A$ on $\msp$ and its abelianization $\bar{\msp}$. We assume that $\msp^{A}$ consist of isolated points.
The fixed point $\lambda \in \msp^{A}$ defines the stabilizing map $\Phi_{\lambda}: A \rightarrow G$, defined by:
$$\Phi_{\lambda}(a) [\lambda] = [\lambda] a$$
for any representative $[\lambda]\in T^{*} V$.

Fix some maximal torus $T\subset G$ such that $\Phi_{\lambda}(A) \subset T $. Let $\Delta$ be the root system and $W$ be the Weyl group corresponding to
choice of $T$. We also fix some split $\Delta=\Delta^+ \cup \Delta^-$ to a positive and negative roots.

Denote by $G_{\lambda}\subset G$, $\Delta_{\lambda} \subset \Delta$, $W_{\lambda}\subset W$ the stabilizer of $\Phi_{\lambda}(A)$, its root system and Weyl group respectively. We denote $\Delta^{\pm}_{\lambda}=\Delta_\lambda \cap \Delta^{\pm}$.

Let $\fL(\lambda) \in \bar{\msp}^{A}$ be some lift of $\lambda$ to the abelianization. Note, that if $x\in \fL(\lambda)$ is some point in the lift then,
the point $w  x$  for $w\in W$ is also in $\fL(\lambda)$. Moreover, $x$ and $w x$ are in different components of $\fL(\lambda)$ if $w\not\in W_\lambda$. In general:
$$
\fL(\lambda) \simeq G_{\lambda}/T\times W/W_\lambda
$$
Therefore, the lift of the fixed points may consist of maximally $|W|$ isolated points. If the stabilizer subgroup is not trivial, i.e. $G_{\lambda} \neq T$, then
the number of component drops to $|W/W_\lambda|$ but some of the components can have larger dimensions.

Let us denote by $ \fL( w {\lambda}) \in \bar{\msp}^{A}$  $w\in |W/W_\lambda|$ one of fixed components. Then we have the stabilizing map:
$$
\bar{\Phi}_{w,\lambda}: A \rightarrow T
$$
and we denote by $\bar{\Phi}^{*}_{w, \lambda}: \ft^{*} \rightarrow \fa^{*}$ the corresponding map on coalgebras.

Let us fix a chamber $\cham \subset \fa$. We say that  $\fL (w {\lambda})$ is the \textit{dominant lift } of a fixed point $\lambda$ with respect to $\cham$
if $\langle\bar{\Phi}^{*}_{w, \lambda}(\alpha),\cham \rangle>0$ for  all roots $\alpha \in \Delta^{+}\setminus \Delta^{+}_{\lambda}$.

\begin{theorem}(\cite{Shenf1})
\label{Shenth}
Let $\msp$ and $\bar{\msp}$ be some hyperk\"{a}hler quotient and its abelianization as above, with a symplectic action of torus $A$.
For a fixed chamber in the Lie algebra $\cham\subset \fa$ denote:
$$
Stab^{G}_{\cham} : H^{\bullet}_{\textbf{G}} ( \msp^{A}  ) \longrightarrow H^{\bullet}_{\textbf{G}} ( \msp ), \ \ \ Stab^{T}_{\cham} : H^{\bullet}_{\textbf{G}} ( \bar{\msp}^{A}  ) \longrightarrow H^{\bullet}_{\textbf{G}} ( \bar{\msp} )
$$
the corresponding stable maps.  Let $\gamma, \delta \in \msp^{A}$ be two fixed points, and $\fL(v \gamma)$ be the dominant lift of the first point. Then we have the following relation among stable maps:
\be
\label{shrf}
\left.Stab^{G}_{\cham}(\gamma)\right|_{\delta}=\sum\limits_{w \in W/W_{\delta}}\, \dfrac{\left.Stab^{T}_{\cham}(  \fL(v \gamma)  )\right|_{ \fL (w \delta) }}{\prod\limits^{ \ \ }_{\alpha\in \Delta^{+}\setminus \Delta^{+}_{\delta}}\bar{\Phi}^{*}_{w,\delta} (\alpha) (\bar{\Phi}^{*}_{w,\delta} (\alpha)+\hbar)}
\ee
\end{theorem}

\subsection{Abelianization of the instanton moduli space}
The instanton moduli space is an example of hyperk\"{a}hler quotient provided by ADHM construction.
Consider the symplectic space:
$$
T^{*}V=\{ (X,Y,I,J) | X,Y\in Hom(\matC^n,\matC^n), I\in Hom(\matC^{r},\matC^n), J\in Hom(\matC^{n},\matC^{r})\}
$$
such that $(X,J)$ is in $V$ and $(Y,I)$ is the symplectic dual part. Consider the  hamiltonian action of $G=GL(n)$  on $T^{*}V$ defined explicitly together with its moment $\mu: T^{*} V \rightarrow \fg $ map by:
\be
g  \cdot (X,Y,I,J) \longrightarrow (g X g^{-1},g Y g^{-1}, g I, J g^{-1})
\ee
\be
\mu(X,Y,I,J)=[X,Y]+I J
\ee
Let $\theta=\det$ is the character of $G$, then the instanton moduli space is isomorphic to the following hyperk\"{a}hler quotient:
\be
\msp(r,m)= T^{*} V /\!\!/_{\theta} G = \mu^{-1}(0)^{\theta-ss}/G
\ee
where $\theta-ss$ denotes $\theta$-semistable points. Let $T\subset G$ is a maximal torus and $\ft=Lie(T)$. Then
we use a projection $\pi: \fg\rightarrow \ft$ to define a $T$ - moment map $\mu_{T}=\pi_{T} \circ \mu$. The abelianization of instanton moduli space is the following $T$ - quotient:
$$
\bar{\msp}(r,m)= T^{*} V /\!\!/_{\theta} T = \mu^{-1}_{T}(0)^{\theta-ss}/T
$$
Let $\{X_{k,m}, J_{k,l}\}$ $k,m=1..n$, $l=1..r$ is a basis of the vector space $V=(X,J)$.
To reduce notations to one in section \ref{hvar}, we denote $T^{m}\simeq (\matC^{*})^{n^2+nr}$ the torus acting on this space by dilating of coordinate vectors. This action extends uniquely to symplectic action on $T^{*}V$. Let $T^{k}=T$ (maximal torus of $GL(n)$) and $T^{d}=T^{m}/T^{k}$. Then, in the notations of section \ref{hvar} we have the following exact sequences:
$$
0\rightarrow \ft^k \stackrel{i}{\rightarrow} \ft^m \stackrel{j}{\rightarrow} \ft^d \rightarrow 0
$$
and its dual:
$$
0\rightarrow (\ft^d)^{*} \stackrel{j^*}{\rightarrow} (\ft^m)^{*} \stackrel{i^*}{\rightarrow} (\ft^k)^{*} \rightarrow 0
$$
i.e. $m=n^2+n r$, $k=n$ and $d=n^2+n(r-1)$.  By  theorem \ref{Pth} the Chern classes $X_{i,j}$ and $J_{k,m}$ generate the cohomology ring
$H^{\bullet}_{\textbf{G}}(\bar{\msp}(r,n) )$.

\subsection{Lifts of the fixed points}
Let us consider several tori acting on the hyperk\"{a}hler quotient $\bar{\msp}(r,n)$. Let $A$ be a maximal torus of $GL(r)$, acting on the elements by
$$
(X,Y,I,J)\rightarrow (X,Y,I g,g^{-1} J)
$$
Consider an action of $(\matC^{*})^2$ by
$$
(z,w) \cdot (X,Y,I,J)\rightarrow (z X, w Y,I,J)
$$
Note that the action of this torus does not fix the symplectic form, and scales it by $z w$.
Let $C=A\times (\matC^{*})^2$. Let $B\subset C$ be a codimension one subtorus fixing the symplectic form.

Thus we have $A\subset B \subset C$, and $B$ is a maximal subtorus of $C$ preserving the symplectic form. The $C$ - equivariant cohomology of
$\msp(r,n)$ is a module over:
$$
\matC[u_1,u_2,...,u_r,t_1,t_2] = H^{\bullet}_{C}( \cdot )
$$
where we denote by $u_1,...,u_r$ the equivariant parameters corresponding to characters of $A$ and $t_1,t_2$ the characters of $C/A$, such that the character of symplectic form is $\hbar=t_1+t_2$.

Let us consider a $r$-tuple of partitions corresponding to a fixed point  $\lb\in \msp^{C}(r,n)$ with $|\lb|=n$.  Let $(X_\lambda,Y_\lambda,I_\lambda,J_\lambda)\in T^{*}V$ denotes its representative. To abuse notations, we say that the vector space $\matC^{n}$ is spanned by boxes $\Box_k$, $k=1...n$. The operators $X_\lambda,Y_\lambda$ are represented by square $n\times n$ matrices indexed by boxes. Consider the fixed component of a lift $\fL(w \lb)$ corresponding to matrices
with nonzero matrix elements $(X_{\lb})_{\Box_1,\Box_2}$ only for  $\varphi_{\Box_2}^{\lb}-\varphi_{\Box_1}^{\lb}=t_1$ and $J_{k,\Box}$ with $\varphi_{\Box_2}^{\lb}=u_k$ and similarly for the dual part $(Y,i)$. Let us consider a general character $\sigma : \matC^{*}\rightarrow C $ be cocharacter defined explicitly by:
$$
\sigma : z \mapsto (z^{u_1},...,z^{u_{r}},z^{t_1},z^{t_2})
$$
To check that the specified operators represent the fixed point it is enough to note that we have a stabilizing map $\Phi_{v,\lb}: \matC^{*} \rightarrow C$ such that
\be
\label{stcon}
\Phi_{v,\lb}(z) X_{\lb}  \Phi_{\lb}^{-1}(z)= z^{t_1} X_{ \lb}, \ \ \ diag(z^{u_1},...,z^{u_r}) J_{\lb} = J_{\lb} \Phi_{v,\lb}^{-1}(z)
\ee
Where the stabilizing map is given explicitly by:
\be
\label{stmap}
\Phi_{v,\lb}(z)=diag(z^{\varphi_{\Box_1}^{\lb}},...,z^{\varphi_{\Box_n}^{\lb}})
\ee
To fix the positive and negative roots of $GL(n)$, we introduce the following natural ordering on boxes: given an $r$-tuple of partitions $\lb$, we turn every partition $\lambda_k$ by $45^{\circ}$ and place them on a line, such  that the boxes have coordinates $\rho_{\Box}^{\lambda}=u_{k}+x(\Box)t_1-y(\Box) t_1$ as in the figure \ref{ordfig}. We also assume that all characters are from one of the two chambers (\ref{chambers}), such that different partitions on the figure \ref{ordfig} "do not intersect".  Given this picture, we order boxes from left to right, and from the bottom to the top, for example in the figure \ref{ordfig} we have a 3-tuple of partitions $\lb=([4,2,1],[4,1,1],[2,2,1])$ and the number in the box correspond to the ordering.

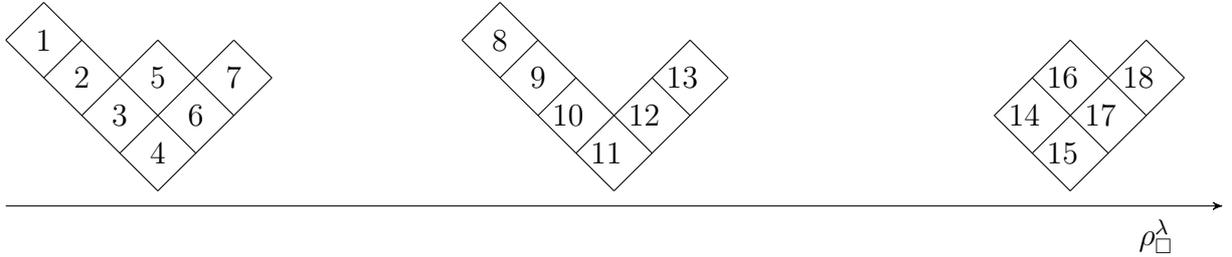
\begin{figure}[ht]
\begin{tikzpicture}
\draw [-] (-7,0) -- (-9,2);
\node [left] at (-6.75,0.5) {$4$};
\node [left] at (-6.75,1.5) {$5$};
\node [left] at (-7.75,1.5) {$2$};
\node [left] at (-5.75,1.5) {$7$};
\node [left] at (-6.25,1) {$6$};
\node [left] at (-7.25,1) {$3$};
\node [left] at (-8.25,2) {$1$};
\draw [-] (-9,2) -- (-8.5,2.5);
\draw [-] (-8,1) -- (-7,2);
\draw [-] (-8.5,1.5) -- (-8,2);
\draw [-] (-6,1) -- (-7,2);
\draw [-] (-5.5,1.5) -- (-6,2);
\draw [-] (-7.5,0.5) -- (-6,2);
\draw [-] (-6.5,0.5) -- (-8.5,2.5);
\draw [-] (-7,0) -- (-5.5,1.5);

\draw [-] (-1,0) -- (-3,2);
\draw [-] (-3,2) -- (-2.5,2.5);
\draw [-] (-2,1) -- (-1.5,1.5);
\draw [-] (-2.5,1.5) -- (-2,2);
\draw [-] (0,1) -- (-0.5,1.5);
\draw [-] (0.5,1.5) -- (0,2);
\draw [-] (-1.5,0.5) -- (0,2);
\draw [-] (-0.5,0.5) -- (-2.5,2.5);
\draw [-] (-1,0) -- (0.5,1.5);

\node [left] at (-0.75,0.5) {$11$};
\node [left] at (-1.75,1.5) {$9$};
\node [left] at (0.25,1.5) {$13$};
\node [left] at (-0.25,1) {$12$};
\node [left] at (-1.25,1) {$10$};
\node [left] at (-2.25,2) {$8$};

\draw [-] (5,0) -- (4,1);
\draw [-] (4,1) -- (5,2);
\draw [-] (6,1) -- (5,2);
\draw [-] (6.5,1.5) -- (6,2);
\draw [-] (4.5,0.5) -- (6,2);
\draw [-] (5.5,0.5) -- (4.5,1.5);
\draw [-] (5,0) -- (6.5,1.5);

\node [left] at (5.25,0.5) {$15$};
\node [left] at (5.25,1.5) {$16$};
\node [left] at (6.25,1.5) {$18$};
\node [left] at (5.75,1) {$17$};
\node [left] at (4.75,1) {$14$};

\draw [->] (-9,-0.2) -- (7,-0.2);

\node [left] at (6.5,-0.6) {$\rho^{\lambda}_{\Box}$};

\end{tikzpicture}
\caption{Box ordering for 3-partition $\lb=([4,2,1],[4,1,1],[2,2,1])$ \label{ordfig}}
\end{figure}

Denote by $h_{s}(\lb)$ the heights of the diagram $\lb$ (considered as in the figure above) defined as the number of boxes with the same coordinate:
\be
\label{hh}
h_{s}(\lb)=\# \{\Box \in \lb|:  \rho_{\Box}^{\lb}=s \}
\ee
such that if $\lb$ has only  heights $h_{s}=1,0$ then it necessarily consist of hook partitions.

Let us now consider the action of torus $B\subset C$. This torus fixes the symplectic form, thus to restrict the previous consideration to $B$ its enough to to substitute $t_{2}=-t_1$. Such that the stabilizing map (\ref{stmap})  takes the form:
\be
\label{ef}
\Phi_{v,\lb}(z)=diag(z^{\rho_{\Box_1}^{\lb}},...,z^{\rho_{\Box_n}^{\lb}})
\ee
The stabilizer $G_{\lb}\subset GL(n)$ is the subgroup commuting with (\ref{ef}), such that from (\ref{hh}) we obtain:
\be
\label{stabg}
G_{\lb}\simeq \prod\limits_{s}\, GL\Big(h_s(\lb) \Big)
\ee

The positive roots are  $\alpha_{\Box_1,\Box_2}: X\rightarrow X_{\Box_1,\Box_1}-X_{\Box_2,\Box_2}$ for $\Box_1<\Box_{2}$ with respect to the chosen ordering.
Thus the set $\Delta^{+}_{\lb}=\Delta_{\lb}\cap\Delta^{+} $ has the form:
\be
\label{psr}
\Delta^{+}_{\lb}=\{ \alpha_{\Box_1,\Box_2} | \Box_1<\Box_{2}, \rho_{\Box_1}^{\lb}=\rho_{\Box_2}^{\lb}  \}
\ee

\begin{lemma}
The specified lift $\fL(v \lb)$ is dominant with respect to chamber $\cham_{+}$:
\end{lemma}
\begin{proof}
Indeed,
$$
\bar{\Phi}_{v,\lb}^{*}(\alpha_{\Box_1,\Box_2})=\rho_{\Box_1}^{\lb}-\rho_{\Box_2}^{\lb}
$$
and the last expression is positive on $\cham_{+}$ for $\alpha_{\Box_1,\Box_2} \in \Delta^{+}\setminus\Delta^{+}_{\lb}$.
\end{proof}
We will also need the following simple lemma:

\begin{lemma}
\label{lem1}
The restriction maps:
$$
H^{2}_{\textbf{G}}(\bar{\msp}(r,n) ) \stackrel{\kappa_{C}}{\longrightarrow} H^{2}_{C}(\bar{\msp}(r,n) )\simeq \fc^{*}, \ \ \  H^{\bullet}_{\textbf{G}}(\bar{\msp}(r,n) ) \stackrel{\kappa_{B}}{\longrightarrow} H^{\bullet}_{B}(\bar{\msp}(r,n) )\simeq\fb^{*}
$$
have the following form:
$$
\kappa_{C}(X_{\Box_1,\Box_2})=\varphi^{\lb}_{\Box_1}-\varphi^{\lb}_{\Box_2}+t_1, \ \ \ \kappa_{C}(J_{k,\Box})=u_k-\varphi^{\lb}_{\Box}
$$
$$
\kappa_{B}(X_{\Box_1,\Box_2})=\rho^{\lb}_{\Box_1}-\rho^{\lb}_{\Box_2}+t_1, \ \ \ \kappa_{B}(J_{k,\Box})=u_k-\rho^{\lb}_{\Box}
$$
\end{lemma}
\begin{proof}
The first line follows from transformation properties (\ref{stcon}). The second line is the restriction of  $C$-character to $B$-character given by $t_2=-t_1$.
\end{proof}
As a consequence we obtain:
\begin{proposition}
\be \label{we}
Stab_{\cham}\Big( \fL(v \lb)  \Big)= \prod\limits_{\Box_1,\Box_2 \in \lb }\, X_{\Box_1,\Box_2}^{* \langle \rho^{\lb}_{\Box_1}-\rho^{\lb}_{\Box_2}+t_1,\cham \rangle}
\prod \limits_{k=1}^{r} \prod\limits_{\Box \in \lb} \, J_{k,\Box}^{*\langle u_m-\rho^{\lb}_{\Box}, \cham \rangle}
\ee
\end{proposition}
\begin{proof}
This follows from the second line of lemma \ref{lem1} and (\ref{stb}).
\end{proof}
Let $w \in W/W_{\lb}$ and $\fL(w\lnu)$ is some component of the lift. In the considered case the Weyl group is the group of permutations $W=S_n$.
The element $w$ acts on the boxes of $\lnu$ by permutations, and boxes of $w \lnu$  are ordered in a way that is obtained from canonical one (described above)
by applying permutation $w$. In this notations we have:
\begin{proposition}
\label{prth}
$$
\left.Stab_{\cham}\Big( \fL(v \lb)  \Big)\right|_{\fL(w \lnu)}= \prod\limits_{\Box_1,\Box_2 \in \lb }\, (\varphi^{w \lnu}_{\Box_1}-\varphi^{w \lnu}_{ \Box_2}+t_1)^{*\langle \rho^{\lb}_{\Box_1}-\rho^{\lb}_{\Box_2}+t_1,\cham \rangle}
\prod \limits_{k=1}^{r} \prod\limits_{\Box \in \lb} \, (u_k-\varphi^{w \lnu}_{\Box})^{*\langle u_m-\rho^{\lb}_{\Box}, \cham \rangle}
$$
\end{proposition}
\begin{proof}
The proof is a direct application of formula for $\kappa_{C}$ from lemma \ref{lem1} to $(\ref{we})$.
\end{proof}

With all this results we can formulate the main theorem of this section:

\begin{theorem}
\label{rstthm}
Let $B\subset C$ are as above. Let $\lb, \lnu \in \msp^{B}(r,n)$ are two fixed points. Then the restrictions of stable envelope $Stab_{\cham_\pm}(\lb) \in H^{\bullet}_{C}(\msp(r,n))$ to other fixed points is given by:
$$
\left.Stab_{\cham_+}(\lb)\right|_{\lnu}=\dfrac{1}{\frak{z} (\lb)}\,\sum\limits_{w \in S_{n}}\dfrac{\prod\limits_{\Box_1,\Box_2 \in \lb }\, (\varphi^{w \lnu}_{\Box_1}-\varphi^{w \lnu}_{ \Box_2}+t_1)^{* \langle\rho^{\lb}_{\Box_1}-\rho^{\lb}_{\Box_2}+t_1|\cham_+\rangle}
\prod \limits_{k=1}^{r} \prod\limits_{\Box \in \lb} \, (u_k-\varphi^{w \lnu}_{\Box})^{*\langle u_m-\rho^{\lb}_{\Box}|\cham_+\rangle}}
{ \prod\limits_{{\Box_1<\Box_2} \atop { \rho^{\lb}_{\Box_1} \neq  \rho^{\lb}_{\Box_1}  }}\,( \varphi^{w \lnu}_{\Box_1}- \varphi^{w\lnu}_{\Box_2}  )\,( \varphi^{w\lnu}_{\Box_1}- \varphi^{w\lnu}_{\Box_2} +\hbar ) }
$$
$$
\left.Stab_{\cham_-}(\lb)\right|_{\lnu}=\dfrac{1}{\frak{z} (\lb)}\,\sum\limits_{w \in S_{n}}\dfrac{\prod\limits_{\Box_1,\Box_2 \in \lb }\, (\varphi^{w \lnu}_{\Box_1}-\varphi^{w \lnu}_{ \Box_2}+t_1)^{* \langle\rho^{\lb}_{\Box_1}-\rho^{\lb}_{\Box_2}+t_1|\cham_-\rangle}
\prod \limits_{k=1}^{r} \prod\limits_{\Box \in \lb} \, (u_k-\varphi^{w \lnu}_{\Box})^{*\langle u_m-\rho^{\lb}_{\Box}|\cham_-\rangle}}
{ \prod\limits_{{\Box_1>\Box_2} \atop { \rho^{\lb}_{\Box_1} \neq  \rho^{\lb}_{\Box_1}  }}\,( \varphi^{w \lnu}_{\Box_1}- \varphi^{w\lnu}_{\Box_2}  )\,( \varphi^{w\lnu}_{\Box_1}- \varphi^{w\lnu}_{\Box_2} +\hbar ) }
$$
with
\be \label{fzet}  \frak{z} (\lb) = \prod\limits_{s} h_{s}(\lb)!\ee
\end{theorem}
\begin{proof}
First, we prove the theorem for $Stab_{\cham_+}$. The proof is a direct application of Shenfeld's theorem \ref{Shenth}.
From the stabilizing map (\ref{stmap}) we obtain:
$$
\bar{\Phi}^{*}_{v,\lb}( \alpha_{\Box_1,\Box_2} )=\varphi^{\lb}_{\Box_1}-\varphi^{\lb}_{\Box_2}
$$
By (\ref{psr}) we have:
$$
\Delta^{+}\setminus\Delta^{+}_{\lb}=\{ \alpha_{\Box_1,\Box_2}| \Box_1<\Box_2, \rho^{\lb}_{\Box_1}\neq \rho^{\lb}_{\Box_2}   \}
$$
The last two formulas give the denominator in (\ref{shrf}), and the numerator is given by proposition \ref{prth}.
Finally, the summand in (\ref{shrf}) is symmetric with the respect to the action of $W_\lb$ thus we have:
$$
\sum\limits_{W/ W_{\lb}} = \dfrac{1}{| W_{\lb}|}\,\sum\limits_{W}
$$
From (\ref{stabg}) we have $W_\lb=\prod\limits_{s} S_{h_{s}} $, and thus $|W_\lb|=\prod\limits_{s} \,h_s (\lb)!$. This finishes the prove for $\cham_+$. Now, to obtain coefficients for the inverse chamber, it is enough to substitute $\cham_+\rightarrow \cham_-$ in the numerator, and replace the product over $\Box_1<\Box_2$ by same product over $\Box_1>\Box_2$. The last, corresponds to taking the opposite order on boxes induced by opposite chamber.
\end{proof}
Note, that the $r=1$ case of this theorem was proved in \cite{Shenf1}.

\subsection{Polynomials representing cohomology classes}
We have the tori $A\subset B\subset C$ acting on the instanton moduli space $\msp(r,n)$ such that
$A$ and $B$ preserve the symplectic form and $C$ dilates it with a character $\hbar=t_1+t_2$.
Remind that the fixed set of $A$ has the form:
$$
\msp^{A}(r,n) = \coprod\limits_{n_1+...+n_r=n}\, Hilb_{n_1}\times...\times Hilb_{n_r}
$$
such that the corresponding equivariant cohomology can be identified with the space of polynomials by the Nakajima construction:
$$
H^{\bullet}_{C}\Big( \coprod \limits_{n=0}^{\infty}\,\msp^{A}(r,n) \Big)=\matC[p^{(1)}_1,p_2^{(1)},....p^{(2)}_1,p_2^{(2)},...p^{(r)}_1,p_2^{(r)}...]\otimes \matC(u_1,...,u_r,t_1,t_2)
$$
The fixed set $\msp^{B}(r,n)=\msp^{C}(r,n)$ is discrete and consist of isolated fixed points labeled by
$r$-tuple of partitions $\lb$ consisting of $|\lb|=n$ boxes.
The generalized Jack polynomials $J_{\lb}$ and $J_{\lb}^{*}$ give two dual bases in this space. In addition we have two dual bases composed of Schur polynomials:
\be \label{schurbas}
s_{\lb}(p_{k}^{(i)})=s_{\lambda_1}(p_{k}^{(1)})...s_{\lambda_r}(p_{k}^{(r)}), \ \ s^{*}_{\lb}(p_{k}^{(i)})=s^{*}_{\lambda_1}(p_{k}^{(1)})...s^{*}_{\lambda_r}(p_{k}^{(r)})
\ee
where  $s_{\lambda}(p_{k})$ are the standard Schur polynomials, and $s^{*}(p_{k})$ their dual with respect to the scalar product  $\langle,\rangle$:
$$
\langle s_{\lambda}, s^{*}_{\mu}\rangle=\delta_{\lambda,\mu}
$$
such that explicitly $s^*_{\lambda}(p_k)=s_{\lambda}(-t_2/t_1 p_k)$.

In the previous sections we fixed  defined opposite chambers $\cham_{\pm} \subset \fb$ and $\cham^{\prime}_{\pm} \subset \fa$.
in this section we will also need the chambers of one-dimensional torus $B/A$ denoted as:
$$
\cham_{+}/\cham^{\prime}_{+}=\{ t_1>0 \}, \ \ \ \cham_{-}/\cham^{\prime}_{-}=\{ t_1<0 \}
$$
Let us consider the stable maps associated with this chambers
$$
Stab_{\cham^{\, \prime}_{\pm}} : H^{\bullet}_{C}( \msp^{A}(r,n) ) \longrightarrow H^{\bullet}_{C}( \msp(r,n) )
$$
$$
Stab_{\cham_{\pm}} : H^{\bullet}_{C}( \msp^{B}(r,n) ) \longrightarrow H^{\bullet}_{C}( \msp(r,n) )
$$
$$
Stab_{\cham_{\pm}/\cham^{\, \prime}_{\pm}} : H^{\bullet}_{C}(  \msp^{B}(r,n)  ) \longrightarrow H^{\bullet}_{C}(  \msp^{A}(r,n)  )
$$
\begin{theorem}
\label{invth}
The basis of the Schur polynomials (\ref{schurbas}) is related to the basis of the generalized Jack polynomials by the following transition matrices:
\be
\label{eqsj}
s_{\lb}=\sum\limits_{|\lnu|=|\lb|}\, U_{\lb,\lnu}, \frac{J_{\lnu}}{E_{\lnu,\lnu}}, \ \ \ s^{*}_{\lb}=\sum\limits_{|\lnu|=|\lb|}\, U^{*}_{\lb,\lnu}, \dfrac{J^{*}_{\lnu}}{E_{\lnu,\lnu}}
\ee
where\begin{small}
$$
  {{U}}_{\bar \lambda,\bar\mu}=\dfrac{1}{\frak{z}({\lb})} \sum\limits_{\sigma\in S_{|\lb|}}\,
\dfrac{\prod\limits_{\Box_1,\Box_2=1}^{|\lm|} \, \Big(\varphi^{\sigma\lm}_{\Box_2} -
 \varphi^{\sigma\lm}_{\Box_1} +t_1 \Big)^{\ast\langle \rho^{\lb}_{\Box_2}-\rho^{\lb}_{\Box_1}+t_1 | \cham_{+}\rangle} \prod\limits_{d=1}^{r}\prod\limits_{\Box=1}^{|\lm|}\Big(u_d -\varphi^{\sigma\lm}_{\Box} \Big)^{\ast \langle u_d-\rho^{\lb}_{\Box} | \cham_{+} \rangle}}
 {\prod\limits_{\Box_1<\Box_2 \atop \rho^{\lb}_{\Box_1}\neq \rho^{\lb}_{\Box_2} } \Big( \varphi^{\sigma\lm}_{\Box_1}-\varphi^{\sigma\lm}_{\Box_2} \Big)
 \Big( \varphi^{\sigma\lm}_{\Box_1}-\varphi^{\sigma\lm}_{\Box_2} +\hbar\Big)}
$$
$$
{{U}}^{*}_{\bar \lambda,\bar\mu}=\dfrac{1}{\frak{z}({\lb})} \sum\limits_{\sigma\in S_{|\lb|}}\,
\dfrac{\prod\limits_{\Box_1,\Box_2=1}^{|\lm|} \, \Big(\varphi^{\sigma\lm}_{\Box_2} -
 \varphi^{\sigma\lm}_{\Box_1} +t_1 \Big)^{\ast\langle \rho^{\lb}_{\Box_2}-\rho^{\lb}_{\Box_1}+t_1 | \cham_{-}\rangle} \prod\limits_{d=1}^{r}\prod\limits_{\Box=1}^{|\lm|}\Big(u_d -\varphi^{\sigma\lm}_{\Box} \Big)^{\ast \langle u_d-\rho^{\lb}_{\Box} | \cham_{-} \rangle}}
 {\prod\limits_{\Box_1>\Box_2 \atop \rho^{\lb}_{\Box_1}\neq \rho^{\lb}_{\Box_2} } \Big( \varphi^{\sigma\lm}_{\Box_1}-\varphi^{\sigma\lm}_{\Box_2} \Big)
 \Big( \varphi^{\sigma\lm}_{\Box_1}-\varphi^{\sigma\lm}_{\Box_2} +\hbar\Big)}
$$
\end{small}
and
\be
E_{\lnu,\lnu}=\prod\limits_{i,j=1}^{r}\, e_{\nu_{i},\nu_{j}}(u_i-u_j)
\ee
is the Euler class of the tangent bundle at $\lnu$.
\end{theorem}
\begin{proof}
We will prove the first half of the theorem for $U_{\lb,\lnu}$, the proof for $U^{*}_{\lb,\lnu}$ is analogous.
Consider the class $Stab_{\cham_{+}} (\lb) \in H^{\bullet}_{C}( \msp(r,n) )$. The expansion of this class in classes of the fixed points $[\lnu]\in H^{\bullet}_{C}( \msp(r,n) )$ has the form:
$$
Stab_{\cham_{+}} (\lb)=\sum\limits_{|\lnu|=|\lb|}\, \dfrac{\left.Stab_{\cham_{+}} (\lb)\right|_{\lnu}}{ \left.e\Big( T \msp(r,n)  \Big)\right|_{\lnu} }\,[\lnu]
$$
where the norm $e\Big( T \msp(r,n)  \Big)$ is the Euler class of the tangent bundle on the instanton moduli space. The denominator of this expression is given by theorem \ref{rstthm}, thus $\left.Stab_{\cham_{+}} (\lb)\right|_{\lnu}=U_{\lb,\lnu}$. Computation of the Euler class gives:
$$
\left.e\Big( T \msp(r,n)  \Big)\right|_{\lnu}=\prod\limits_{i,j=1}^{r}\, e_{\nu_{i},\nu_{j}}(u_i-u_j)=E_{\lnu,\lnu}
$$
Now, the theorem follows from the identities:
$$
Stab_{\cham^{\, \prime}_{+}}\Big( J_{\lb} \Big)=[\lb], \ \ \ Stab_{\cham^{\, \prime}_{+}}\Big( s_{\lb} \Big)=Stab_{\cham_{+}}\Big(\lb \Big)
$$
and from the fact that the localized map $Stab_{\cham^{\, \prime}_{+}}$ is an isomorphism. The first identity is the definition of generalized Jack polynomials. For the second one, we note that we have the following factorization of stable maps (see lemma 3.6.1 in \cite{Maulik:2012wi}):
$$
Stab_{\cham_{+}}=Stab_{\cham^{\, \prime}_{+}} \circ Stab_{\cham_{+}/\cham^{\, \prime}_{+}}
$$
The stable envelope of a point on the Hilbert schemes is given by Schur polynomial $Stab_{\cham_{+}/\cham^{\, \prime}_{+}}(\lb)=s_{\lb}$ \cite{Shenf1}, and the theorem follows.
\end{proof}
Now, the proof of theorem \ref{thmone} is elementary.
\begin{proof}
From (\ref{eqsj}) we obtain $T=E U^{-1}$, where $E=diag(E_{\lb,\lb})$ is a diagonal matrix.
From the scalar product of the first equation and the second equations of (\ref{eqsj}), using the fact $\langle s_{\lb}, s_{\lm} \rangle=\delta_{\lb,\lm}$ and $\langle J_{\lb}, J^*_{\lm} \rangle=\delta_{\lb,\lm} E_{\lb,\lb}$ we obtain:
$$
1=U E^{-1} U^{* t}
$$
where $L^{t}$ denotes transposed matrix. Thus we obtain $T_{\lb,\lnu}=U^{*}_{\lnu,\lb}$. Same consideration applies to $T^{*}_{\lm,\lnu}$. Now theorem follows from explicit formulas of the previous theorem.
\end{proof}

\section{Appendix A: Stable Map \label{sten}}
\setcounter{equation}{0}
\def\theequation{A.\arabic{equation}}
In this section, following \cite{Maulik:2012wi}, we recall the definition of \textit{the stable map} playing important role in the this paper.

Assume that a pair of algebraic tori  $A\subset T$ acts on the symplectic variety $X$. This action induces the action on $H^{0}(\Omega^2_{X})$. Assume that the induced action of $T$ on $H^{2}(X)$ scales the symplectic form $\omega$. It implies that the one-dimensional
subspace ${\mathbb{C}}\omega\subset H^{2}(X)$ is a subrepresentation of $T$. We denote by $\hbar$ its character. Assume, that the action of the smaller torus $A$ preserves $\omega$.

Our goal in this section is to describe the natural map defined in \cite{Maulik:2012wi}:
$$
\textrm{Stab}_{\cham}: H^{\bullet}_{T}(X^{A}) \rightarrow H^{\bullet}_{T}(X)
$$
depending on chamber $\cham$ in the Lie algebra $a_{\mathbb{R}}$.  For a fixed cycle $\gamma \in H^{\bullet}_{T}(X^{A})$ the element
$\textrm{Stab}_{\cham}(\gamma)\in H^{\bullet}_{T}(X)$ is called the \textit{stable envelope} of $\gamma$.


\subsection{Chamber decomposition \label{cdec}}

Let $A\simeq ({\mathbb{C}}^{\ast})^r $ be an algebraic torus of rank $r$. Let
\be
\label{coc}
 c(A)=\{\, A \rightarrow {\mathbb{C}}^{\ast} \,\}\simeq{{\mathbb{Z}}}^r, \ \ \ t(A)=\{\, {\mathbb{C}}^{\ast} \rightarrow A \,\}\simeq{{\mathbb{Z}}}^r
\ee
be the group of characters and cocharacters respectively.  We define the real part of the Lie algebra and its dual as:
\be
\label{rlie}
a_{{{\mathbb{R}}}}=t(A)\otimes_{{\mathbb{Z}}}{\mathbb{R}}\simeq{\mathbb{R}}^r\subset \textrm{Lie}(A), \ \ \ a_{{{\mathbb{R}}}}^{\ast}=c(A)\otimes_{{\mathbb{Z}}}{\mathbb{R}}\simeq{\mathbb{R}}^r
\ee
The natural pairing $t(A)\times c(A)\rightarrow {{\mathbb{Z}}}$ linearly extends to the pairing for the real Lie algebra:
$$\langle \ \ , \ \ \rangle : a_{{{\mathbb{R}}}} \times a_{{{\mathbb{R}}}}^{\ast} \rightarrow {{\mathbb{R}}}$$

\noindent
\begin{definition}
Let $X^{A}$  be the fixed set of $A$. The normal bundle $N$ to $X^A$ in $X$ has a natural structure of an $A$-module and  splits to the direct sum of complex, one-dimensional, irreducible components. The subset $\Delta\subset a_{\mathbb{R}}^{\ast}$ consisting of the characters appearing in $N$ is called \textit{root system } of $A$.
\end{definition}

A weight $\alpha \in a_{{{\mathbb{R}}}}^{\ast}$ defines a hyperplane in $a_{{{\mathbb{R}}}}$:
\be
\textrm{ker}_\alpha=\{v\in a_{{{\mathbb{R}}}}: \langle \alpha, v \rangle =0\}
\ee
The  hyperplanes corresponding to the roots partition $a_{{{\mathbb{R}}}}$ into the set of open chambers:
\be
a_{{\mathbb{R}}} \setminus \bigcup\limits_{\alpha \in \Delta}\, \textrm{ker}_{\alpha }=\coprod\limits_{i} \cham_{i}
\ee

The  hyperplanes $\textrm{ker}_{\alpha}$, clearly,  are walls of the chambers. In general, the walls $\textrm{ker}_{\alpha}$ define a stratification of the space $a_{{\mathbb{R}}}$ by the chain of sets: the set of points that do not lie on any wall (these are chambers), the set of points lying on exactly one wall, the set of points lying on the intersection of two walls and so on.

The stratification of $a_{{\mathbb{R}}}$ encodes the information about the $A$-action on $X$. Indeed, consider a cocharacter $\sigma: {\mathbb{C}}^{\ast} \rightarrow A$. It defines certain ${\mathbb{C}}^{\ast}$-action on $X$. If $\sigma$ does not belong to some wall i.e. is inside one of the chambers then,
${\mathbb{C}}^{\ast}$ - action has the same set of the fixed points $X^{{\mathbb{C}}^{\ast}} = X^{A}$. Assume now, that $\sigma$ lies on exactly one wall $\ker_{\alpha}$. Then torus ${\mathbb{C}}^{\ast}$ acts trivially on the component of the normal bundle corresponding to the character $\alpha$. Therefore,  ${\mathbb{C}}^{\ast}$ preserves the corresponding direction in $X$ and the fixed set $X^{{\mathbb{C}}^{\ast}}$ gets larger then $X^{A}$.   Thus, the stratification of $a_{{{\mathbb{R}}}}$ by the walls $\ker_{\alpha}$ corresponds to the types of the fixed sets $X^{{\mathbb{C}}^{\ast}}$ arising from different choice of the cocharacters  $\sigma: {\mathbb{C}}^{\ast} \rightarrow A$. In the extreme case $\sigma=0$, corresponding to the intersection of all walls, we have $X^{{\mathbb{C}}^{\ast}}=X$.
\subsection{Stable leaves and slopes}
\noindent
\begin{definition}
Let us fix some chamber $\cham \subset a_{\mathbb{R}}$ and let $\sigma\in \cham$ be a cocharacter. We say that the point $x \in X$ is $\cham$\textit{-stable} if the following limit exists:
\be
\lim_{\cham} x\stackrel{\textrm{def}}{=} \lim\limits_{z\rightarrow 0} \sigma(z)\cdot x \in X^{A}
\ee
\end{definition}
This definition does not depend on the choice of $\sigma$ in the chamber $\cham$, which explains the notation $\lim\limits_{\cham}$.
For a component $Z$ of the fixed set $ X^{A}$ we define its  \textit{stable leaf} as the set of stable points "attracting" to $Z$:

\noindent
\begin{definition}
\be
\textrm{Leaf}_{\cham}(Z)=\{ x|\lim_{\cham} x  \in Z \}
\ee
\end{definition}
The choice of a chamber $\cham$ defines a partial order on the components $Z\subset X^{A}$. We say that:
 $$Z_1 \succeq Z_{2} \ \  \Leftrightarrow \ \ \overline{\textrm{Leaf}_{\cham}(Z_1)} \cap Z_2\neq \emptyset.$$ Using this ordering we define the  \textit{stable slope} of a component  $Z\subset X^{A}$ as follows:

\noindent
\begin{definition}
\be
\textrm{Slope}_{\cham}(Z)=\coprod\limits_{Z^{\prime} \preceq\, Z} \textrm{Leaf}_{\cham}(Z^{\prime}).
\ee
\end{definition}
\noindent

\subsection{Polarization}
 Let $Z\subset X^{A}$ be a component of the fixed set. The choice of a chamber~$\cham$ gives a decomposition of the normal bundle to $Z$ in $X$ into the weight spaces that are positive or negative on the chamber $\cham$:
$$
N_{Z}=N_{+}\oplus N_{-}
$$
Remind, that by our assumption the action of $A$ preserves the symplectic form $\omega$ and $T$ scales it with character $\hbar$. Thus, we have:
\be
\label{polar}
(N_{+})^{\vee}=N_{-}\otimes \hbar
\ee
where for convenience we denoted by the same symbol $\hbar$  the trivial $T$-equivariant line bundle over $Z$ with the action of $T$ on its fiber corresponding to the character $\hbar$.

Assume that $\alpha_{i}$, $i=1...\textrm{codim}(Z)/2$ are the weights of $N_{+}$ then,  the $A$-weights of $N_{-}$ are given by $(-\alpha_{i})$. Therefore, the $A$-equivariant Euler class of $N_{Z}$ (with a sign) is a perfect square:
\be
\label{Ec}
\varepsilon^2=(-1)^{\textrm{codim}(Z)/2} e(N_Z)=\prod\limits_{i=1}^{\textrm{codim}(Z)/2} \alpha_{i}^2\ \
\ee

\noindent
\begin{definition}
The \textit{polarization} of $Z$ is a formal choice of a sign in the square root of (\ref{Ec}):
\be
\left.\varepsilon\right|_{H^{\bullet}_{A}(\textrm{pt})}=\pm \prod\limits_{i=1}^{\textrm{codim}(Z)/2} \alpha_{i}
\ee
We say that the sign $\pm e(N_{-})\in {H^{\bullet}_{T}(Z)}$ is chosen according to the polarization if it restricts to $\varepsilon$ in $H^{\bullet}_{A}(Z)$.
\end{definition}
\subsection{Stable envelope \label{se}}
\label{stabth}
The stable envelope is defined by the following theorem.

\begin{theorem} Under assumption above, there exists a unique map of $H_{T}^{\bullet}(\textrm{pt})$ modules:
$$
\textrm{Stab}_{\cham, \varepsilon} : H^{\bullet}_{T}(X^{A}) \rightarrow H_{T}^{\bullet}(X)
$$
depending on the choice of chamber $\cham$ and polarization $\varepsilon$. For a component $Z\subset X^{A}$ and any $\gamma \in H^{\bullet}_{T}(Z)$ the stable envelope $\Gamma=\textrm{Stab}_{\cham,\varepsilon}(\gamma)$ is defined uniquely by the following properties:
\begin{itemize}
\item
$\textrm{supp}(\Gamma)\subset \textrm{Slope}_{\cham}(Z)$
\item $\left.\Gamma\right|_{Z}=\pm e(N_{-}) \cup \gamma$ with the sign chosen according to the polarization~$\varepsilon$.
\item $\deg_{A} \left.\Gamma\right|_{Z^{\prime}}< \textrm{codim}(Z^{\prime})/2$,  for any $Z^{\prime}> Z$
\end{itemize}
\end{theorem}
The proof of this theorem can be found in \cite{Maulik:2012wi}. In addition, the description of $\textrm{Stab}_{\cham, \varepsilon}$
as the Lagrangian correspondence can be found there.

As we mentioned above, the choice of polarization is a formality corresponding to choice of signs. We will use symbol $\textrm{Stab}_{\cham}$ for the stable map meaning that some polarization $\varepsilon$ is chosen.

\section{Appendix B: Some explicit formulae}
\setcounter{equation}{0}
\def\theequation{B.\arabic{equation}}
In this appendix we would like to write down examples of $J_{\lb}$ for several first values of $r$ and $n$. This examples are computed using the theorem \ref{thmone}, which allows perform a computations for relatively high number of boxes. The explicit formulae, however, become hairy very fast.
\subsection{Case $r=1$}
In this case the theorem  \ref{thmone} gives the expansion of the standard Jack polynomials corresponding to the parameter $\beta=-t_2/t_1$ in Schur polynomials.

\vspace{3mm}
\noindent
$
J_{{[1]}}=t_{{2}}s_{{[1]}}\\
J_{{[1,1]}}=2\,{t_{{2}}}^{2}s_{{[1,1]}}\\
J_{{[2]}}= \left( t_{{1}}+t_{{2}} \right) t_{{2}}s_{{[1,1]}}- \left( t
_{{1}}-t_{{2}} \right) t_{{2}}s_{{[2]}}\\
J_{{[1,1,1]}}=6\,{t_{{2}}}^{3}s_{{[1,1,1]}} \\
J_{{[2,1]}}=2\,{t_{{2}}}^{2} \left( t_{{1}}+t_{{2}} \right) s_{{[1,1,1
]}}-{t_{{2}}}^{2} \left( t_{{1}}-2\,t_{{2}} \right) s_{{[2,1]}}\\
J_{{[3]}}=t_{{2}} \left( 2\,t_{{1}}+t_{{2}} \right)  \left( t_{{1}}+t_
{{2}} \right) s_{{[1,1,1]}}-2\,t_{{2}} \left( t_{{1}}+t_{{2}} \right)
 \left( t_{{1}}-t_{{2}} \right) s_{{[2,1]}}+t_{{2}} \left( t_{{1}}-t_{
{2}} \right)  \left( 2\,t_{{1}}-t_{{2}} \right) s_{{[3]}}
$
\subsection{Case $r=2$}
In this case the fixed points are labeled by a pair of partitions and the theorem \ref{thmone} gives:

\vspace{3mm}
\noindent
$
J_{{[],[1]}}= \left( t_{{1}}+t_{{2}}-u_{{1}}+u_{{2}} \right) t_{{2}}s_
{{[],[1]}}\\
J_{{[1],[]}}= \left( t_{{1}}+t_{{2}} \right) t_{{2}}s_{{[],[1]}}-t_{{2
}} \left( u_{{1}}-u_{{2}} \right) s_{{[1],[]}}\\
J_{{[],[1,1]}}=2\,{t_{{2}}}^{2} \left( t_{{1}}+t_{{2}}-u_{{1}}+u_{{2}}
 \right)  \left( t_{{1}}+2\,t_{{2}}-u_{{1}}+u_{{2}} \right) s_{{[],[1,
1]}}\\
J_{{[],[2]}}=t_{{2}} \left( 2\,t_{{1}}+t_{{2}}-u_{{1}}+u_{{2}}
 \right)  \left( t_{{1}}+t_{{2}}-u_{{1}}+u_{{2}} \right)  \left( t_{{1
}}+t_{{2}} \right) s_{{[],[1,1]}}- \\
t_{{2}} \left( t_{{1}}-t_{{2}}
 \right)  \left( 2\,t_{{1}}+t_{{2}}-u_{{1}}+u_{{2}} \right)  \left( t_
{{1}}+t_{{2}}-u_{{1}}+u_{{2}} \right) s_{{[],[2]}}\\
J_{{[1],[1]}}= \left( 2\,t_{{1}}+t_{{2}}-u_{{1}}+u_{{2}} \right) {t_{{
2}}}^{2} \left( t_{{1}}+t_{{2}} \right) s_{{[],[1,1]}}+{t_{{2}}}^{2}
 \left( t_{{2}}-u_{{1}}+u_{{2}} \right)  \left( t_{{1}}+t_{{2}}
 \right) s_{{[],[2]}}+\\
 {t_{{2}}}^{2} \left( t_{{2}}-u_{{1}}+u_{{2}}
 \right)  \left( u_{{2}}-u_{{1}}+t_{{1}} \right) s_{{[1],[1]}}\\
J_{{[1,1],[]}}=2\,{t_{{2}}}^{2} \left( t_{{1}}+t_{{2}} \right)
 \left( t_{{1}}+2\,t_{{2}}-u_{{2}}+u_{{1}} \right) s_{{[],[1,1]}}-2\,{
t_{{2}}}^{2} \left( t_{{1}}+t_{{2}} \right)  \left( u_{{1}}-u_{{2}}
 \right) s_{{[],[2]}}-\\
 2\,{t_{{2}}}^{2} \left( t_{{1}}+t_{{2}} \right)
 \left( u_{{1}}-u_{{2}} \right) s_{{[1],[1]}}+2\,{t_{{2}}}^{2} \left(
u_{{1}}-u_{{2}} \right)  \left( t_{{2}}-u_{{2}}+u_{{1}} \right) s_{{[1
,1],[]}}\\
J_{{[2],[]}}= \left( t_{{1}}+t_{{2}} \right) t_{{2}} \left( 2\,t_{{1}}
u_{{1}}-
2\,t_{{1}}u_{{2}}+2\,{t_{{1}}}^{2}+3\,t_{{2}}t_{{1}}+{t_{{2}}}
^{2} \right) s_{{[],[1,1]}}-\\
\left( t_{{1}}+t_{{2}} \right) t_{{2}}
 \left( 2\,t_{{1}}u_{{1}}-2\,t_{{1}}u_{{2}}+2\,{t_{{1}}}^{2}-t_{{2}}t_
{{1}}-
{t_{{2}}}^{2} \right) s_{{[],[2]}}-\\
2\,{t_{{2}}}^{2} \left( t_{{1
}}+
t_{{2}} \right)  \left( u_{{1}}-u_{{2}} \right) s_{{[1],[1]}}+\\
t_{{2
}} \left( t_{{1}}+t_{{2}} \right)  \left( u_{{1}}-u_{{2}}+t_{{1}}
 \right)  \left( u_{{1}}-u_{{2}} \right) s_{{[1,1],[]}}-t_{{2}}
 \left( t_{{1}}-t_{{2}} \right)  \left( u_{{1}}-u_{{2}}+t_{{1}}
 \right)  \left( u_{{1}}-u_{{2}} \right) s_{{[2],[]}}
$
\subsection{Case $r=3$}
In this case we recover explicit formulae obtained in  \cite{Mironov:2013oaa}:

\vspace{3mm}
\noindent
$
J_{{[],[],[1]}}=t_{{2}} \left( t_{{1}}+t_{{2}}-u_{{1}}+u_{{3}}
 \right)  \left( t_{{1}}+t_{{2}}-u_{{2}}+u_{{3}} \right) s_{{[],[],[1]
}}\\
J_{{[],[1],[]}}=t_{{2}} \left( t_{{1}}+t_{{2}} \right)  \left( t_{{1}}
+t_{{2}}-u_{{1}}+u_{{2}} \right) s_{{[],[],[1]}}-t_{{2}} \left( t_{{1}
}+t_{{2}}-u_{{1}}+u_{{2}} \right)  \left( u_{{2}}-u_{{3}} \right) s_{{
[],[1],[]}}\\
J_{{[1],[],[]}}= \left( t_{{1}}+t_{{2}} \right)  \left( u_{{1}}-u_{{2}
}+t_{{1}}+t_{{2}} \right) t_{{2}}s_{{[],[],[1]}}-t_{{2}} \left( t_{{1}
}+t_{{2}} \right)  \left( u_{{1}}-u_{{3}} \right) s_{{[],[1],[]}}+\\t_{{
2}} \left( u_{{1}}-u_{{2}} \right)  \left( u_{{1}}-u_{{3}} \right) s_{
{[1],[],[]}}
$
\bibliographystyle{utphys}
\bibliography{bib}
\
\end{document}